\documentclass[a4paper,USenglish,cleveref, autoref, thm-restate]{lipics-v2019}
\usepackage{amsmath,verbatim,amssymb,amsfonts,amscd,graphicx}
\usepackage{amsthm}
\usepackage{enumerate}
\usepackage{mathrsfs}
\usepackage{mathtools}
\usepackage{cite}
\usepackage{appendix}
\usepackage{hyperref}
\usepackage{commath}
\usepackage{proba}
\usepackage{xspace}
\usepackage{tikz}
\usetikzlibrary{shapes}
\usetikzlibrary{decorations.pathreplacing}

\usepackage{cite}

\newtheorem{construction}[theorem]{Construction}

\title{A New Connection Between Node and Edge Depth Robust Graphs}
\titlerunning{A New Connection Between Node and Edge D.R. Graphs} 
\author{Jeremiah Blocki}{Department of Computer Science, Purdue University, West Lafayette, USA}{jblocki@purdue.edu}{}{}
\author{Mike Cinkoske}{Department of Computer Science, University of Illinois Urbana-Champaign, USA}{mjc18@illinois.edu}{}{Supported by NSF REU \#1934444. } 
\authorrunning{J. Blocki and M. Cinkoske} 
\Copyright{Jeremiah Blocki and Mike Cinkoske} 



\category{} 

\funding{This research was supported by the National Science Foundation under awards CNS \#1755708 and CNS \#1704587.}

\nolinenumbers 

\newcommand{\ignore}[1]{}

\newcommand{\COMMENTED}[1]{{}}

\newcommand{\eps}{\epsilon}

\newcommand{\mdef}[1]{{\ensuremath{#1}}\xspace}  

\newcommand{\node}[2]{\mdef{\langle#1, #2\rangle}}   

\def \indeg    {\mdef{\mathsf{indeg}}}

\newcommand{\namedref}[2]{\hyperref[#2]{#1~\ref*{#2}}}

\newcommand{\corlab}[1]{\label{cor:#1}}
\newcommand{\corref}[1]{\namedref{Corollary}{cor:#1}}

\newenvironment{proofof}[1]{\noindent {\em Proof of #1.  }}{$\hfill\square$}

\newenvironment{remindertheorem}[1]{\medskip \noindent {\bf Reminder of  #1.  }\em}{}

\EventEditors{James R. Lee}
\EventNoEds{1}
\EventLongTitle{12th Innovations in Theoretical Computer Science Conference (ITCS 2021)}
\EventShortTitle{ITCS 2021}
\EventAcronym{ITCS}
\EventYear{2021}
\EventDate{January 6--8, 2021}
\EventLocation{Virtual Conference}
\EventLogo{}
\SeriesVolume{185}
\ArticleNo{69}

\begin{document}
\maketitle


\begin{abstract}
Given a directed acyclic graph (DAG) $G = (V,E)$, we say that $G$ is $(e,d)$-depth-robust (resp. $(e,d)$-edge-depth-robust) if for any set $S \subset V$ (resp. $S \subseteq E$) of at most $|S| \leq e$ nodes (resp. edges) the graph $G-S$ contains a directed path of length $d$. While edge-depth-robust graphs are potentially easier to construct many applications in cryptography require node depth-robust graphs with small indegree. We create a graph reduction that transforms an $(e, d)$-edge-depth-robust graph with $m$ edges into a $(e/2,d)$-depth-robust graph with $O(m)$ nodes and constant indegree. One immediate consequence of this result is the first construction of a provably $(\frac{n \log \log n}{\log n}, \frac{n}{(\log n)^{1 + \log \log n}})$-depth-robust graph with constant indegree, where previous constructions for $e =\frac{n \log \log n}{\log n}$ had $d = O(n^{1-\epsilon})$. Our reduction crucially relies on ST-Robust graphs, a new graph property we introduce which may be of independent interest. We say that a directed, acyclic graph with $n$ inputs and $n$ outputs is $(k_1, k_2)$-ST-Robust if we can remove any $k_1$ nodes and there exists a subgraph containing at least $k_2$ inputs and $k_2$ outputs such that each of the $k_2$ inputs is connected to all of the $k_2$ outputs. If the graph if $(k_1,n-k_1)$-ST-Robust for all $k_1 \leq n$ we say that the graph is maximally ST-robust. We show how to construct maximally ST-robust graphs with constant indegree and $O(n)$ nodes. Given a family $ \mathbb{M}$ of ST-robust graphs and an arbitrary $(e, d)$-edge-depth-robust graph $G$ we construct a new constant-indegree graph $  \mathrm{Reduce}(G, \mathbb{M})$ by replacing each node in $G$ with an ST-robust graph from $ \mathbb{M}$. We also show that ST-robust graphs can be used to construct (tight) proofs-of-space and (asymptotically) improved wide-block labeling functions. 
\end{abstract}

\section{Introduction}

Given a directed acyclic graph (DAG) $G = (V,E)$, we say that $G$ is $(e,d)$-reducible (resp. $(e,d)$-edge reducible) if there is a subset $S\subseteq V$ (resp. $S \subseteq E$) of $|S|\leq e$ nodes (resp. edges) such that $G-S$ does not contain a directed path of length $d$. If a graph is not $(e,d)$-reducible (resp. $(e,d)$-edge reducible) we say that the graph is $(e,d)$-depth robust (resp. $(e,d)$-edge-depth-robust). Depth robust graphs have found many applications in the field of cryptography in the construction of proofs of sequential work~\cite{ITCS:MahMorVad13}, proofs of space \cite{C:DFKP15,ITCS:Pietrzak19a}, and in the construction of data independent memory hard functions (iMHFs). For example, highly depth robust graphs are known to be necessary \cite{C:AlwBlo16} and sufficient \cite{EC:AlwBloPie17} to construct iMHF's with high amortized space time complexity. While edge depth-robust graphs are often easier to construct~\cite{Sch83}, most applications require node depth-robust graphs with small indegree.

It has been shown \cite{Val77} that in any DAG with $m$ edges and $n$ nodes and any $i \leq \log_2 n$, there exists a set $S_i$ of $\frac{mi}{\log n}$ edges that will destroy all paths of length $n/2^i$ forcing $\mathrm{depth}(G - S_i) \leq \frac{n}{2^i}$. For DAGs with constant indegree we have $m = O(n)$ edges so an equivalent condition holds for node depth robustness~\cite{C:AlwBlo16}, since a node can be removed by removing all the edges incident to it. In particular, there exists a set $S_i$ of $O(\frac{ni}{\log n})$ nodes such that  $\mathrm{depth}(G - S_i) \leq \frac{n}{2^i}$ for all $i < \log n$. It is known how to construct an $(c_1n / \log n, c_2n)$-depth-roubst graph, for suitable $c_1, c_2 > 0$ \cite{EC:AlwBloPie17} and an $(c_3n, c_4n^{1-\epsilon})$-depth-robust graph for small $\epsilon$ for  \cite{Sch83}. 

An open challenge is to construct constant indegree $(c_1 n i/\log n, c_2n/2^i)$-depth-robust graphs which match the Valiant bound~\cite{Val77} for intermediate values of $i = \omega(1)$ and $i = o(\log n)$. For example, when $i = \log \log n$ then the Valiant bound~\cite{Val77} does not rule out the existence of $(c_1 n i/\log n, c_2n/\log n)$-depth-robust graphs with constant indegree. Such a graph would yield asymptotically stronger iMHFs~\cite{C:BHKLXZ19}. While there are several constructions that are conjectured to be $(c_1 n i/\log n, c_2n/\log n)$-depth-robust the best provable lower bound for $(e=c ni/\log n,d)$-depth robustness of a constant indegree graph is $d = \Omega(n^{1-\eps})$. For edge-depth robustness we have constructions of graphs with $m = O(n \log n)$ edges which are $(e_i,d_i)$-edge depth robust for any $i$ with $e_i = mi/\log n$ and $d_i = n/\log^{i+1} n$ --- much closer to matching the Valiant bound~\cite{Val77}.

\subsection{Contributions}
Our main contribution is a graph reduction that transforms any $(e, d)$-edge-depth-robust graph with $m$ edges into a $(e/2, d)$-depth-robust graph with $O(m)$ nodes and constant indegree. Our reduction utilizes ST-Robust graphs, a new graph property we introduce and construct. We believe that ST-Robust graphs may be of independent interest. 

 Intuitively, a $(k_1, k_2)$-ST-Robust graph with $n$ inputs $I$ and $n$ outputs $O$ satisfies the property that, even after deleting $k_1$ nodes from the graph we can find $k_2$ inputs $x_1,\ldots,x_{k_2}$ and $k_2$ outputs $y_1,\ldots,y_{k_2}$ such that {\em every} input $x_i$ ($i \in[k_2]$) is still connected to {\em every} output $y_j$ ($j \in[k_2]$). If we can guarantee that the each directed path from $x_i$ to $y_j$ has length $d$ then we say that the graph is $(k_1, k_2,d)$-ST-Robust. A maximally depth-robust graph should be $(k_1,n-k_1)$ -depth robust for any $k_1$.

\begin{definition}{\bf{ST-Robust}} \label{def:st}
Let $G = (V, E)$ be a DAG with $n$ inputs, denoted by set $I$ and $n$ outputs, denoted by set $O$. Then $G$ is  $(k_1, k_2)$-ST-robust if $\forall D \subset V(G)$ with $|D| \leq k_1$, there exists subgraph $H$ of $G - D$ with $|I \cap V(H) | \geq k_2$ and $|O \cap V(H) | \geq k_2$ such that $\forall s \in I \cap V(H)$ and $\forall t \in O \cap V(H)$ there exists a path from $s$ to $t$ in $H$. If $\forall s \in I \cap V(H)$ and $\forall t \in O \cap V(H)$ there exists a path from $s$ to $t$ of length $\geq d$ then we say that $G$ is $(k_1,k_2,d)$-ST-robust.
\end{definition}

\begin{definition}{\bf{Maximally ST-Robust}}
Let $G = (V, E)$ be a constant indegree DAG with $n$ inputs and $n$ outputs. Then $G$ is $c_1$-maximally ST-robust (resp. $c_1$ max ST-robust with depth $d$) if there exists a constant $0 < c_1 \leq 1$ such that $G$ is $(k, n-k)$-ST-robust (resp. $(k,n-k,d)$-ST-robust) for all $k$ with $0 \leq k \leq c_1n$. If $c_1 = 1$, we just say that $G$ is maximally ST-Robust. 
\end{definition}

We show how to construct  maximally ST-robust graphs with constant indegree and $O(n)$ nodes and we show how maximally ST-robust graphs can be used to transform any $(e, d)$-edge-depth-robust graph $G$ with $m$ edges into a $(e/2, d)$-depth-robust graph $G'$ with $O(m)$ nodes and constant indegree. Intuitively, in our reduction each node $v \in V(G)$ with degree $\delta(v)$ is replaced with a maximally ST-robust graph $M_{\delta(v)}$ with $\delta(v)$ inputs/outputs. Incoming edges into $v$ are redirected into the inputs $I_{\delta(v)}$ of the ST-robust graph. Similarly, $v$'s outgoing edges are redirected out of the outputs $O_{\delta(v)}$ of the ST-robust graph. Because $M_{\delta(v)}$ is maximally ST-robust, when a node is removed from $M_{\delta(v)}$ the set of inputs and outputs where each input connects to every output has at most one input and one output node removed. Each input or output node removed from $M_{\delta(v)}$ corresponds to removing at most one edge from the original graph. Thus, removing $k$ nodes from $M_{\delta(v)}$ corresponds to destroying at most $2k$ edges in the original graph $G$. 

 Our reduction gives us a fundamentally new way to design node-depth-robust graphs: design an edge-depth-robust graph (easier) and then reduce it to a node-depth-robust graph. The reduction can be used with a construction from \cite{Sch83} to construct a $(\frac{n \log \log n}{\log n}, \frac{n}{(\log n) ^{1 + \log \log n}})$-depth-robust graph. We conjecture that several prior DAG constructions (e.g, \cite{EGS75,Sch83,EC:AlwBloPie18}) are actually $(n \log \log n, \frac{n}{\log n})$-edge-depth-robust. If any of these conjectures are true then our reduction would immediately yield the desired $(\frac{n \log \log n}{\log n}, \frac{n}{\log n})$-depth-robust graph.

We also present several other applications for maximally ST-robust graphs including (tight) proofs-of-space and wide block-labeling functions.


\section{Edge to Node Depth-Robustness}

In this section, we use the fact that linear sized, constant indegree, maximally ST-robust graphs exist to construct a transformation of an $(e,d)$-edge-depth robust graph with $m$ edges into an $(e,d)$-node-depth robust graph with constant indegree and $O(m)$ nodes. In the next section we will construct a family of ST-robust graphs that satisfies Theorem \ref{theorem:STRobust}.

\begin{theorem} \label{theorem:STRobust} {\bf (Key Building Block)}
  There exists a family of graphs  $\mathbb{M} = \{M_n\}_{n=1}^\infty$ with the property that for each $n \geq 1$, $M_n$ has constant indegree, $O(n)$ nodes, and is maximally ST-Robust.
\end{theorem}

\subsection{Reduction Definition}

Let $G = (V, E)$ be a DAG, and let $\mathbb{M}$ be as in Theorem \ref{theorem:STRobust}. Then we define Reduce($G$, $\mathbb{M}$) in construction \ref{con:reduce} as follows:

\begin{construction}[\bf{Reduce(G, $\mathbb{M}$)}] \label{con:reduce}
  Let $G = (V, E)$ and let $\mathbb{M}$ be the family of graphs defined above. For each $M_n \in \mathbb{M}$, we say that $M_n = (V(M_n), E(M_n))$, with $V(M_n) = I(M_n) \cup O(M_n) \cup D(M_n)$, where $I(M_n)$ are the inputs of $M_n$, $O(M_n)$ are the outputs, and $D(M_n)$ are the internal vertices. For $v \in V$, let $\delta(v) = \max\{\mathrm{indegee}(v), \mathrm{outdegree}(v) \}$ Then we define Reduce($G$)$\ = (V_R, E_R)$, where $V_R = \{(v, w) | v \in V, w \in V(M_{\delta(v)}) \}$ and $E_R = E_{internal} \cup E_{external}$. We let $E_{internal} = \{((v, u'), (v, w')) | v \in V, (u', w') \in E(M_{\delta(v)}) \}$. Then for each  $v \in V$, we define an $In(v) = \{u ~:~(u,v) \in E\}$ and $Out(v) = \{u~:~(v,u) \in E\}$ and then pick two injective mappings $\pi_{in,v}: In(v) \rightarrow I(V(M_{\delta(v)}))$ and  $\pi_{out,v}: Out(v) \rightarrow O(V(M_{\delta(v)}))$. We let $E_{external} = \{ ((u, \pi_{out, u}(v)),(v, \pi_{in, v}(u))) ~:~ (u,v)\in E \}$.
 
  Intuitively, to costruct Reduce($G$, $\mathbb{M}$) we replace every node of $G$ with a constant indegree, maximally ST-robust graph, mapping the edges connecting two nodes from the outputs of one ST-robust graph to the inputs of another.   Then for every $e = (u, w) \in E$, add an edge from an output of $M_{\delta(u)}$ to an input of $M_{\delta(w)}$ such that the outputs of $M_{\delta(u)}$ have outdegree at most 1, and the inputs of  $M_{\delta(w)}$ have indegree at most 1. If $v \in V$ is replaced by $M_{\delta(v)}$, then we call $v$ the \textit{genesis node} and $M_{\delta(v)}$ its  \textit{metanode}.
\end{construction}

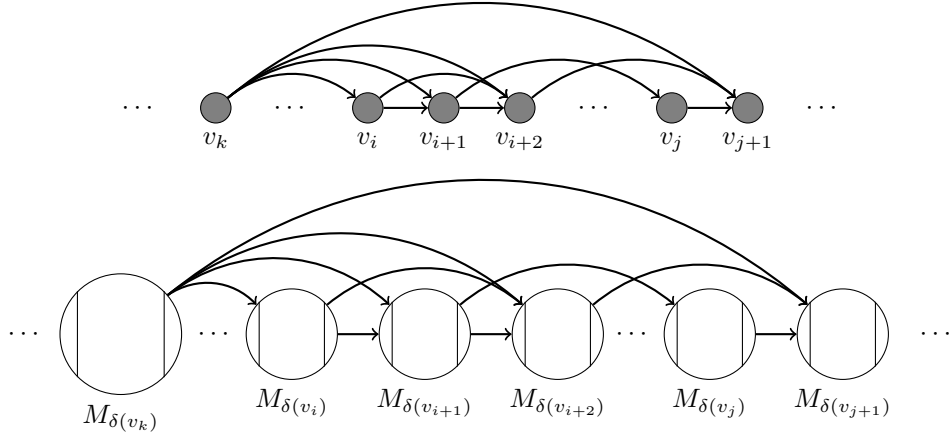
\begin{figure}[t]
  \centering
  \label{transform}
  \begin{tikzpicture}
    [place/.style={circle, draw=black, fill=gray, minimum size=4mm},
      transform/.style={circle, draw=black, fill=none, minimum size=12mm},
            transformLarge/.style={circle, draw=black, fill=none, minimum size=16mm},
    arr/.style={<-, thick}]

    \node[draw=none] (e1) at (-2,0) {$\cdots$};

    \node[place] (vk) at (-1, 0) [label=below:$v_k$] {};
    \node[draw=none] (e1) at (0,0) {$\cdots$};
    \node[place] (vi0) at (1,0) [label=below:$v_i$] {}
    edge [arr,bend right=40] node[auto] {} (vk);

    \node[place] (vi1) at (2,0) [label=below:$v_{i+1}$] {}
    edge [arr] node[auto] {} (vi0)
    edge [arr,bend right=40] node[auto] {} (vk);

    \node[place] (vi2) at (3,0) [label=below:$v_{i+2}$] {}
    edge [arr] node[auto] {} (vi1)
    edge [arr, bend right=40] node[auto] {} (vi0)
    edge [arr,bend right=40] node[auto] {} (vk);

    \node[draw=none] at (4,0) {$\cdots$};
    \node[place] (vj0) at (5,0) [label=below:$v_{j}$] {}
    edge [arr, bend right=40] node[auto] {} (vi1);
    \node[place] (vj1) at (6,0) [label=below:$v_{j+1}$] {}
    edge [arr] node[auto] {} (vj0)
    edge [arr, bend right=40] node[auto] {} (vi2)
    edge [arr,bend right=40] node[auto] {} (vk);

    \node[draw=none] at (7,0) {$\cdots$};

    \node[draw=none] (e1) at (-3.5,-3) {$\cdots$};

     \node[transformLarge] (Vk) at (-2.25, -3) [label=below:$M_{\delta(v_k)}$] {};
    \pgfpathmoveto{\pgfpointanchor{Vk}{north east}}
    \pgfpathlineto{\pgfpointanchor{Vk}{south east}}
    \pgfusepath{draw}
    
    \pgfpathmoveto{\pgfpointanchor{Vk}{north west}}
    \pgfpathlineto{\pgfpointanchor{Vk}{south west}}
    \pgfusepath{draw}
    
    \node[draw=none] (e1) at (-1,-3) {$\cdots$};

    \node[transform] (Vi0) at (0, -3) [label=below:$M_{\delta(v_i)}$] {}
    edge [arr, bend right=40] node[auto] {} (Vk);
    \pgfpathmoveto{\pgfpointanchor{Vi0}{north east}}
    \pgfpathlineto{\pgfpointanchor{Vi0}{south east}}
    \pgfusepath{draw}
    
    \pgfpathmoveto{\pgfpointanchor{Vi0}{north west}}
    \pgfpathlineto{\pgfpointanchor{Vi0}{south west}}
    \pgfusepath{draw}

    \node[transform] (Vi1) at (1.75, -3) [label=below:$M_{\delta(v_{i+1})}$] {} 
    edge [arr] node[auto] {} (Vi0)
    edge [arr,bend right=40] node[auto] {} (Vk);

    \pgfpathmoveto{\pgfpointanchor{Vi1}{north east}}
    \pgfpathlineto{\pgfpointanchor{Vi1}{south east}}
    \pgfusepath{draw}
    
    \pgfpathmoveto{\pgfpointanchor{Vi1}{north west}}
    \pgfpathlineto{\pgfpointanchor{Vi1}{south west}}
    \pgfusepath{draw}

    \node[transform] (Vi2) at (3.5, -3) [label=below:$M_{\delta(v_{i+2})}$] {}
    edge [arr] node[auto] {} (Vi1)
    edge [arr, bend right=40] node[auto] {} (Vi0)
    edge [arr,bend right=40] node[auto] {} (Vk);

    \pgfpathmoveto{\pgfpointanchor{Vi2}{north east}}
    \pgfpathlineto{\pgfpointanchor{Vi2}{south east}}
    \pgfusepath{draw}
    
    \pgfpathmoveto{\pgfpointanchor{Vi2}{north west}}
    \pgfpathlineto{\pgfpointanchor{Vi2}{south west}}
    \pgfusepath{draw}

    \node[draw=none] (e1) at (4.5,-3) {$\cdots$};

    \node[transform] (Vj0) at (5.5, -3) [label=below:$M_{\delta(v_j)}$] {}
        edge [arr, bend right=40] node[auto] {} (Vi1);
    \pgfpathmoveto{\pgfpointanchor{Vj0}{north east}}
    \pgfpathlineto{\pgfpointanchor{Vj0}{south east}}
    \pgfusepath{draw}
    
    \pgfpathmoveto{\pgfpointanchor{Vj0}{north west}}
    \pgfpathlineto{\pgfpointanchor{Vj0}{south west}}
    \pgfusepath{draw}

    \node[transform] (Vj1) at (7.25, -3) [label=below:$M_{\delta(v_{j+1})}$] {}
    edge [arr] node[auto] {} (Vj0)
    edge [arr, bend right=40] node[auto] {} (Vi2)
    edge [arr,bend right=40] node[auto] {} (Vk);

    \pgfpathmoveto{\pgfpointanchor{Vj1}{north east}}
    \pgfpathlineto{\pgfpointanchor{Vj1}{south east}}
    \pgfusepath{draw}
    
    \pgfpathmoveto{\pgfpointanchor{Vj1}{north west}}
    \pgfpathlineto{\pgfpointanchor{Vj1}{south west}}
    \pgfusepath{draw}
        \node[draw=none] (e1) at (8.5,-3) {$\cdots$};

    
  \end{tikzpicture}
  \caption{Diagram of the transformation Reduce($G, \mathbb{M}$) }
\end{figure}

\subsection{Proof of Main Theorem}


We now state the main result of this section which says that if $G$ is edge-depth robust then $\mathrm{Reduce}(G, \mathbb{M})$ is node depth-robust.

\newcommand{\thmReductionEdgeToNode}{  Let $G$ be an (e, d)-edge-depth-robust DAG with $m$ edges. Let $\mathbb{M}$ be a family of max ST-Robust graphs with constant indegree. Then $G' = (V', E') = \mathrm{Reduce}(G, \mathbb{M})$ is $(e/2, d)$-depth robust. Furthermore, $G'$ has maximum indegree  $\max_{v \in V(G)}\{ \mathrm{indeg}\left(M_{\delta(v)}\right)\}$, and its number of nodes is $\sum_{v \in V(G)} \left|V(M_{\delta(v)})\right|$ where $\delta(v) = \max\{\mathrm{indeg}(v),\mathrm{outdeg}(v)\}$.
}
\begin{theorem}
  \label{general}
\thmReductionEdgeToNode
\end{theorem}

A formal proof can be found in Appendix \ref{apdx:MissingProof}. We briefly outline the intuition for this proof below. 

\begin{proof} (Intuition)  The first thing we node is that each graph $M_{\delta(v)}$ has constant indegree at most $c \delta(v)$ nodes for some constant $c>0$. Therefore, the graph $G'$ has $\sum_{v \in V(G)} \left|V( M_{\delta(v)})\right| \leq c \sum_v \delta(v) \leq 2cm$ nodes and $G'$ has constant indegree.

Now for any set $S \subseteq V'$ of nodes we remove from $G'$ we will map $S$ to a corresponding set $S_{irr} \subseteq E$ of at most $|S_{irr}| \leq 2|S|$ irrepairable edges in $G$. We then prove that any path $P$ in $G-S_{irr}$ corresponds to a longer path $P'$ in $G' - S$ that is at least as long. Intuitively, each  incoming edge $(u,v)$ (resp. outgoing edge $(v,w)$) in $E(G)$ corresponds to an input node (resp. output node) in $v$'s corresponding metanode $M_{\delta(v)}$ which we will label $x_{u,v}$ (resp. $y_{v,w}$). If $S \subseteq V'$ removes at most $k$ nodes from the metanode $M_{\delta(v)}$ then, by maximal ST-robustness, we still can find $\delta(v)-k$ inputs and $\delta(v)-k$ outputs that are {\em all} pairwise connected. If $x_{u,v}$ (resp. $y_{v,w}$) is not part of this pairwise connected subgraph then we will add the corresponding edge $(u,v)$ (resp. $(v,w)$) to the set $S_{irr}$.  Thus, the set $S_{irr}$ will have size at most $2|S|$  Claim \ref{claim:irr-bound} in the appendix).

Intuitively, any path $P$ in $G-S_{irr}$ can be mapped to a longer path $P'$ in $G'-S$ (Claim \ref{claim:paths}). If $P$ contains the edges $(u,v),(v,w)$ then we know that the input node $x_{u,v}$ and output node $y_{u,v}$ node in $M_{\delta(v)}$ are still connected in $G'-S$.  
\end{proof}

\begin{corollary}(of Theorem \ref{general})
If there exists some constants $c_1, c_2$, such that we have a family $\mathbb{M} = \{M_n\}_{n=1}^\infty $ of linear sized $|V(M_n)| \leq c_1 n$, constant indegree $\mathrm{indeg}(M_n) \leq c_2$, and maximally ST-Robust graphs, then  $\mathrm{Reduce}(G, \mathbb{M})$ has maximum indegree $c_2$ and the number of nodes is at most $2c_1m$.
\end{corollary}

The next corollary states that if we have a family of maximally ST-Robust graphs with $\mathbb{M} = \{M_k\}_{k=1}^\infty$ depth $d_k$ then we can transform any $(e, d)$-edge-depth-robust DAG $G=(V,E)$ with maximum degree $\delta = \max_{v \in V} \delta(v)$ into $(e/2,d \cdot  d_\delta)$-depth robust graph. Instead of replacing each node $v \in G$ with a copy of $M_{\delta(v)}$, we instead replace each node with a copy of $M_{\delta, v} := M_\delta$, attaching the edges same way as in Construction \ref{con:reduce}. Thus the transformed graph $G'$ has $|V(G)|\times \left| M_{\delta}\right|$ nodes and constant indegree. Intuitively, any path $P$ of length $d$ in $G-S_{irr}$  now maps to a path $P'$ of length $d \times d_\delta$ --- if $P$ contains the edges $(u,v),(v,w)$ then we know that the input node $x_{u,v}$ and output node $y_{u,v}$ node in $M_{\delta,v}$ are connected  in $G'-S$ {\em by a path of length at least} $d_\delta$.  

\newcommand{\thmCorrolaryofDRone}{(of Theorem \ref{general}) Suppose that there exists a family $\mathbb{M} = \{M_k\}_{k=1}^\infty$ of max ST-Robust graphs with depth $d_k$ and constant indegree. Given any $(e, d)$-edge-depth-robust DAG $G$ with $n$ nodes and maximum degree $\delta$ we can construct a DAG $G'$ with $n \times \left| M_{\delta}\right|$ nodes and constant indegree that is $(e/2,d\cdot d_\delta)$-depth robust.}
\begin{corollary} \corlab{DRone} \thmCorrolaryofDRone
\end{corollary}

\begin{proof}(sketch)
  Instead of replacing each node $v \in G$ with a copy of $M_{\delta(v)}$, we instead replace each node with a copy of $M_{\delta, v} := M_\delta$, attaching the edges same way as in Construction \ref{con:reduce}. Thus the transformed graph $G'$ has $|V(G)|\times \left| M_{\delta}\right|$ nodes and constant indegree. Let $S \subset V(G')$ be a set of nodes that we will remove from $G'$. By Claim \ref{claim:paths}, there exists a path $P$ in $G' - S$ that passes through $d$ metanodes $M_{\delta, v_1},\ldots, M_{\delta, v_d}$. Since $M_\delta$ is maximally ST-robust with depth $d_\delta$ the sub-path $P_i = P \cap M_{\delta,v_i}$ through each metanode has length $|P_i| \geq d_\delta$. Thus, the total length of the path is at least $\sum_i |P_i| \geq d\cdot d_\delta$.
\end{proof}

\begin{corollary} (of Theorem \ref{general}) Let $\epsilon > 0$ be any fixed constant. Given any family $\{G_m\}_{m=1}^\infty$ of $(e_m, d_m)$-edge-depth-robust DAGs $G_m$ with $m$ nodes and maximum indegree $\delta_m$ then for some constants $c_1,c_2  > 0$ we can construct a family $\{H_m\}_{m=1}^\infty$ of DAGs such that each DAG $H_m$ is $(e_m/2,d_m \cdot \delta_m^{1-\epsilon})$-depth robust, $H_m$ has maximum indegree at most $c_2$ (constant) and at most $\left| V(H_m)\right| \leq c_1 m \delta_m$ nodes.
\end{corollary}
\begin{proof}(sketch) This follows immediately from \corref{DRone} and from our construction of a family $\mathbb{M}_{\epsilon} = \{M_{k,\epsilon}\}_{k=1}^\infty$ of max ST-Robust graphs with depth $d_k > k^{1-\epsilon}$ and constant indegree. 
\end{proof}

\begin{corollary} (of Theorem \ref{general})
  Let $\{ e_m \}_{m=1}^\infty$ and $\{ d_m \}_{m=1}^\infty$ be any sequence. If there exists a family $\{ G_m \}_{m=1}^\infty$ of $(e_m, d_m)$-edge-depth-robust graphs, where each DAG $G_m$ has $m$ edges, then there is a corresponding family $\{ H_n \}_{n=1}^\infty$ of constant indegree DAGs such that each $H_n$ has $n$ nodes and is $(\Omega(e_n), \Omega(d_n))$-depth-robust. 
\end{corollary}

The original Grate's construction \cite{Sch83}, $G$, has $N = 2^n$ nodes and $m = n2^n$ edges and for any $s \leq n$, and is $(s2^n, \frac{N}{\sum^s_{j=0}\binom{n}{j}})$-edge-depth-robust. For node depth-robustness we only had matching constructions when $s=O(1)$~\cite{EC:AlwBloPie17,CCS:AlwBloHar17} and $s= \Omega( \log N)$~\cite{Sch83} --- no comparable lower bounds were known for intermediate $s$. 

\begin{corollary} (of Theorem \ref{general})
There is a family of constant indegree graphs $\{G_n\}$ such that $G_n$ has $O\left(N=2^n\right)$ nodes and $G_n$ is $(sN/(2n), \frac{N}{\sum^s_{j=0}\binom{n}{j}})$-edge-depth-robust for any $1 \leq s \leq \log n$
\end{corollary}

In particular, setting $s = \log \log n$ and applying the indegree reduction from Theorem \ref{general}, we see that the transformed graph $G'$ has constant indegree,  $N' = O(n2^n)$ nodes, and is $(\frac{N' \log \log N'}{\log N'}, \frac{N'}{(\log N') ^{1 + \log \log N'}})$-depth-robust. Blocki et al. \cite{C:BHKLXZ19} showed that if there exists a node depth robust graph with $e = \Omega(N \log \log N / \log N)$ and $d = \Omega(N \log \log N / \log N)$ then one can obtain another constant indegree graph with pebbling cost $\Omega(N^2 \log \log N / \log N)$ which is optimal for constant indegree graphs. We conjecture that the graphs in \cite{EGS75} are sufficiently edge depth robust to meet these bounds after being transformed by our reduction.

\section{ST Robustness}

In this section we show how to construct maximally ST-robust graphs with constant indegree and linear size. We first introduce some of the technical building blocks used in our construction including  superconcentrators~\cite{Valiant:1976:GPC:1739937.1740084,Pippenger77,GabGal81} and grates~\cite{Sch83}. Using these building blocks we then provide a randomized construction of a $c_1$-maximally ST-robust DAG with linear size and constant indegree for some constant $c_1 >0$  --- sampled graphs are  $c_1$-maximally ST-robust DAG with high probability. Finally, we use $c_1$-maximally ST-robust DAGs to construct a family of maximally ST-robust graphs with linear size and constant indegree. 

\subsection{Technical Ingredients} 

We now introduce other graph properties that will be useful for constructing ST-robust graphs. 

\subsubsection*{Grates}
A DAG $G = (V,E)$ with $n$ inputs $I$ and $n$ outputs $O$ is called a $(c_0,c_1)$-grate if for any subset $S \subset V$ of size $|S| \leq c_0n$ at least $c_1 n^2$ input output pairs $(x,y) \in I \times O$ remain connected by a directed path from $x$ to $y$ in $G-S$. Schnitger \cite{Sch83} showed how to construct $(c_0,c_1)$-grates with $O(n)$ nodes and constant indegree for suitable constants $c_0,c_1>0$.   The notion of an maximally ST-robust graph is a strictly stronger requirement since there is no requirement on which pairs are connected. However, we show that a slight modification of Schnitger's \cite{Sch83} construction is a $(cn,n/2)$-ST-robust for a suitable constant $c$. We then transform this graph into a $c_1$-maximally ST-robust graph by sandwiching it in between two superconcentrators. Finally, we show how to use several $c_1$-maximally ST-robust graphs to construct a maximally ST-robust graph.

\subsubsection*{Superconcentrators}
We say that a directed acyclic graph $G = (V, E)$ with $n$ input vertices and $n$ output vertices is an \textbf{\textit{n}-superconcentrator} if for any $r$ inputs and any $r$ outputs, $1 \leq r \leq n$, there are $r$ vertex-disjoint paths in $G$ connecting the set of these $r$ inputs to these $r$ outputs. We note that there exists linear size, constant indegree superconcentrators~\cite{Valiant:1976:GPC:1739937.1740084,Pippenger77,GabGal81} and we use this fact throughout the rest of the paper. For example, Pippenger \cite{Pippenger77} constructed an  \textit{n}-superconcentrator with at most $41n$ vertices and indegree at most $16$. 

\subsubsection*{Connectors}
We say that  an $n$-superconcentrator is an  \textbf{\textit{n}-connector} if it is possible to specify which input is to be connected to which output by vertex disjoint paths in the subsets of $r$ inputs and $r$ outputs. Connectors and superconcentrators are potential candidates for ST-robust graphs because of their highly connective properties. In fact, we can prove that any  connectors  \textbf{\textit{n}-connector} is maximally ST-robust --- the proof of Theorem \ref{connector-thm} can be found in the appendix. While we have constructions of \textbf{\textit{n}-connector} graphs these graphs have  $O(n \log n)$ vertices and indegree of 2, an information theoretic technique of Shannon\cite{shannon} can be used to prove that any $n$-connector with constant indegree requires {\em at least} $\Omega(n \log n)$ vertices --- see discussion in the appendix. Thus, we cannot use $n$-connectors to build linear sized ST-robust graphs. However, Shannon's information theoretic argument does not rule out the existence of linear size ST-robust graphs.

\newcommand{\connectorIsSTRobust}{If $G$ is an $n$-connector, then $G$ is $(k, n-k)$-ST-robust, for all $1 \leq k \leq n$.}
\begin{theorem}
  \label{connector-thm}
\connectorIsSTRobust
\end{theorem}

\subsection{Linear Size ST-robust Graphs}
ST-robust graphs have similar connective properties to connectors, so a natural question to ask is whether ST-robust graphs with constant indegree require $\Omega(n \log n)$ vertices. In this section, we show that linear size ST-robust graphs exist by showing that a modified version of the Grates construction \cite{Sch83} becomes $c$-maximally ST-robust when sandwiched between two superconcentrators for some constant $c$.

In the proof of Theorem A in \cite{Sch83}, Schnitger constructs a family of DAGs $(H_n | n \in N)$ with constant indegree $\delta_H$, where $n$ is the number of nodes and $H_n$ is $(cn, n^{2/3})$-depth-robust, for suitable constant $c > 0$. We construct a similar graph $G_n$ as follows:

\begin{construction}[$G_n$] \label{ThreeGrates}
 We begin with $H^1_n$, $H_n^2$ and $H_n^3$, three isomorphic copies of $H_n$ with disjoint vertex sets $V_1$, $V_2$ and $V_3$. For each top vertex $v \in V_3$ sample $\tau$ vertices $x_1^v,\ldots,x_\tau^v$ independently and uniformly at random from $V_2$ and for each $i \leq \tau$ add each directed edge $(x_i^v,v)$ to $G_n$ to connect  each of these sampled nodes to $v$. Similarly, for each node vertex $u \in V_2$ sample $\tau$ vertices $x_1^u,\ldots,x_\tau^u$ from $V_1$ independently and uniformly at random and add each directed edge $(x_i^u,u)$ to $G_n$. Note that $\indeg(G_n) \leq \indeg(H_n)+\tau$.
\end{construction}
Schnitger's construction only utilizes two isomorphic copies of $H_n$ and the edges connecting $H^1_n$ and $H^2_n$ a sampled by picking $\tau$ random permutations. In our case the analysis is greatly simplified by picking the edges uniformly and we will need three layers to prove ST-robustness. We will use the following lemma from the Grates paper as a building block. A proof of Lemma \ref{paths-lemma} is included in the appendix  for completeness.

\newcommand{\pathsLemma}{ For some suitable constant $c>0$ any any subset $S$ of $cn/2$ vertices of $G_n$ the graph $H_n^1 - S$ contains $k = cn^{1/3}/2$ vertex disjoint paths $A_1, \ldots , A_k$ of length $n^{2/3}$ and $H_n^2 - S$ contains $k$ vertex disjoint paths $B_1, \ldots, B_k$ of the same length.}

\begin{lemma} \cite{Sch83}
  \label{paths-lemma} \pathsLemma
\end{lemma}

We use Lemma   \ref{paths-lemma} to help establish our main technical Lemma   \ref{prob}. We sketch the proof of Lemma \ref{prob} below. A formal proof can be found in Appendix \ref{apdx:MissingProof}.

\newcommand{\thmProb}{
Let $G_n$ be defined as in Construction \ref{ThreeGrates}. Then for some constants $c > 0$, with high probability $G_n$ has the property that for all $S \subset V(G_n)$ with $|S| = cn/2$ there exists $A \subseteq V(H_n^1)$ and $B\subseteq V(H_n^3)$ such that for every pair of nodes $u \in A$ and $v \in B$ the graph $G_n-S$ contains a path from $u$ to $v$ and $|A|,|B| \geq 9c n/40$.
}

\begin{lemma}
  \label{prob}
  \thmProb
\end{lemma}

\begin{proof}
 (Sketch) Fixing any $S$ we can apply Lemma \ref{paths-lemma} to find $k:=cn^{1/3}/2$ vertex disjoint paths $P_{1,S}^i, \ldots, P_{k,S}^i$ in  $H_n^i$ of length $n^{2/3}$ for each $i \leq 3$. Here, $c$ is the constant from Lemma \ref{paths-lemma}.  Let $U^i_{j,S}$ be the upper half of the $j$-th path in $H_n^i$ and $L_{j,S}^i$ be the lower half and define the event $BAD_{i,S}^{upper}$ to be the event that there exists at least $k/10$ indices $j \leq k$ s.t., $U^2_{j,S}$ is disconnected from $L^3_{i,S}$. We construct $B$ by taking the union of all of upper paths $U^3_{i,S}$  in $H_n^3$ for each non-bad (upper) indices $i$. Similarly, we define $BAD_{i,S}^{lower}$ to be the event that there exists at least $k/10$ indices $j \leq k$ s.t. $U^1_{i,S}$ is disconnected from $L^2_{j,S}$ and we construct $A$ be taking the union of all of the lower paths $L^1_{i,S}$ in $H_n^1$ for each non-bad (lower) indices $i$. We can now argue that any pair of nodes $u \in A$ and $v \in B$ is connected by invoking the pigeonhole principle i.e., if $u \in L^1_{i,S}$ and $v \in U^3_{i',S}$ for good indices $i$ and $i'$ then there exists some path $P_j^2$ in the middle layer $H_n^2$ which can be used to connect $u$ to $v$. We still need to argue that $|A|,|B| \geq c n/3$ for some constant $c$. To lower bound $|B|$  we introduce the event    $ BAD_{S} = |\{i\ :\ BAD_{i,S}^{upper}\}| > \frac{k}{10}$ and note that  unless $BAD_S$ occurs we have $|B| \geq (9k/10) n^{2/3}/2 = 9cn/40$. Finally, we show that $\P[ BAD_{S}]$ is very small and then use union bounds to show that, for a suitable constant $\tau$, the probability $\P[\exists S BAD_{S}]$ becomes negligibly small. A symmetric argument can be used to show that $|A| \geq 9cn/40$. 
\end{proof}

We now use $G_n$ to construct $c$-maximally ST-robust graphs with linear size.

\begin{construction}[$M_n$] \label{con:st}
  We construct the family of graphs $M_n$ as follows: Let the graphs $SC^1_n$ and $SC^2_n$ be linear sized $n$-superconcentrators  with constant indegree $\delta_{SC}$ \cite{Pippenger77}, and let $H^1_n$, $H^2_n$ and $H^3_n$ be defined and connected as in $G_n$,  where every output of $SC^1_n$ is connected to a node in $H^1_n$, every node of $H^3_n$ is connected to an input of $SC^2_n$. 
\end{construction}

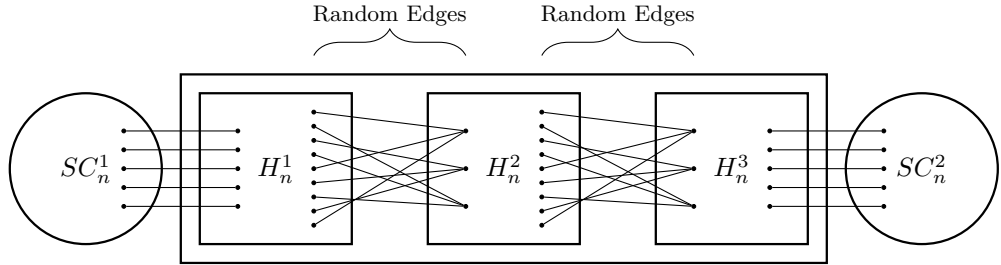
\begin{figure}[h!]
  \centering
  \label{str}
  \begin{tikzpicture}
    [place/.style={circle, draw=black, fill=black,inner sep=0pt, minimum size=.5mm}]

    \draw (1, 1) circle (1cm) [line width=.3mm];
    \draw (1,1) node {$SC^1_n$};
    
    \node[place] (a0) at (1.5, .5) {};
    \node[place] (a1) at (1.5, .75) {};
    \node[place] (a2) at (1.5, 1) {};
    \node[place] (a3) at (1.5, 1.25) {};
    \node[place] (a4) at (1.5, 1.5) {};

    \draw (2.25, -.25) rectangle (10.75 , 2.25) [line width=.3mm];
    
    \draw (2.5, 0) rectangle (4.5, 2) [line width=.3mm];
    \draw (3.5, 1) node {$H^1_n$};
    
    \node[place] (b0) at (3, .5) {}
    edge [-] node[auto] {} (a0);
    \node[place] (b1) at (3, .75) {}
    edge [-] node[auto] {} (a1);
    \node[place] (b2) at (3, 1) {}
    edge [-] node[auto] {} (a2);
    \node[place] (b3) at (3, 1.25) {}
    edge [-] node[auto] {} (a3);
    \node[place] (b4) at (3, 1.5) {}
    edge [-] node[auto] {} (a4);

     \draw [decorate,decoration={brace,amplitude=10pt}] 
(4,2.5) -- (6,2.5) node [black,midway, yshift=15pt]
{\footnotesize Random Edges};
    
    \node[place] (l0) at (6, .5) {};
    \node[place] (l1) at (6, 1) {};
    \node[place] (l2) at (6, 1.5) {};

    \node[place] (r0) at (4, .25) {}
    edge [-] node[auto] {} (l2);
    \node[place] (r1) at (4, .4375) {}
    edge [-] node[auto] {} (l1);
    \node[place] (r2) at (4, .625) {}
    edge [-] node[auto] {} (l0);
    \node[place] (r3) at (4, .8125) {}
    edge [-] node[auto] {} (l1);
    \node[place] (r4) at (4, 1) {}
    edge [-] node[auto] {} (l2);
    \node[place] (r5) at (4, 1.1875) {}
    edge [-] node[auto] {} (l0);
    \node[place] (r6) at (4, 1.375) {}
    edge [-] node[auto] {} (l1);
    \node[place] (r7) at (4, 1.5625) {}
    edge [-] node[auto] {} (l0);
    \node[place] (r8) at (4, 1.75) {}
    edge [-] node[auto] {} (l2);

    \draw (5.5, 0) rectangle (7.5, 2) [line width=.3mm];
    \draw (6.5, 1) node {$H^2_n$};

    \node[place] (n0) at (9, .5) {};
    \node[place] (n1) at (9, 1) {};
    \node[place] (n2) at (9, 1.5) {};

    \node[place] (t0) at (7, .25) {}
    edge [-] node[auto] {} (n2);
    \node[place] (t1) at (7, .4375) {}
    edge [-] node[auto] {} (n1);
    \node[place] (t2) at (7, .625) {}
    edge [-] node[auto] {} (n0);
    \node[place] (t3) at (7, .8125) {}
    edge [-] node[auto] {} (n1);
    \node[place] (t4) at (7, 1) {}
    edge [-] node[auto] {} (n2);
    \node[place] (t5) at (7, 1.1875) {}
    edge [-] node[auto] {} (n0);
    \node[place] (t6) at (7, 1.375) {}
    edge [-] node[auto] {} (n1);
    \node[place] (t7) at (7, 1.5625) {}
    edge [-] node[auto] {} (n0);
    \node[place] (t8) at (7, 1.75) {}
    edge [-] node[auto] {} (n2);
     \draw [decorate,decoration={brace,amplitude=10pt}] 
(7,2.5) -- (9,2.5) node [black,midway, yshift=15pt]
{\footnotesize Random Edges};
    \draw (8.5, 0) rectangle (10.5, 2) [line width=.3mm];
    \draw (9.5, 1) node {$H^3_n$};
    
    \node[place] (c0) at (10, .5) {};
    \node[place] (c1) at (10, .75) {};
    \node[place] (c2) at (10, 1) {};
    \node[place] (c3) at (10, 1.25) {};
    \node[place] (c4) at (10, 1.5) {};

    \node[place] (d0) at (11.5, .5) {}
    edge [-] node[auto] {} (c0);
    \node[place] (d1) at (11.5, .75) {}
    edge [-] node[auto] {} (c1);
    \node[place] (d2) at (11.5, 1) {}
    edge [-] node[auto] {} (c2);
    \node[place] (d3) at (11.5, 1.25) {}
    edge [-] node[auto] {} (c3);
    \node[place] (d4) at (11.5, 1.5) {}
    edge [-] node[auto] {} (c4);

    \draw (12, 1) circle (1cm) [line width=.3mm];
    \draw (12,1) node {$SC^2_n$};
    
  \end{tikzpicture}  
  \caption{A diagram of the constant indegree, linear sized, ST-robust graph $M_n$.}
\end{figure}

\begin{theorem} \label{is-str} 
  There exists a constant $c' > 0$ such that for all sets $S \subset V(M_n)$ with $|S| \leq c'n/2$, $M_n$ is $(|S|, n - |S|)$-ST-robust, with $n$ inputs and $n$ outputs and constant indegree.
\end{theorem}

\begin{proof}
Let $c' = 9c / 40$, where $c$ is the constant from $G_n$. Consider $M_n - S$. Then because $|S \cap (H_n^1 \cup H_n^2)| \leq |S| \leq c' n/2 \leq cn/2$, by Lemma \ref{prob} with a high probability there exists sets $A$ in $H^1_n$ and $B$ in $H^3_n$ with $|A|,|B| \geq \frac{9}{10} k \frac{n^{2/3}}{2} = \frac{9}{40} c n = c'n$, such that every node in $A$ connects to every node in $B$. By the properties of superconcentrators, the size of the set $BAD_1$ of inputs $u$ in $SC^1_n$ that can't reach {\em any} node in $A$ in $M_n-S$.  We claim that $|BAD_1| \leq |S| \leq c'n$. Assume for contradiction that  $|BAD_1|>|S|$ then $SC^1_n$ contains at least $\min\{ |BAD_1|, |A|\} > |S|$ node disjoint paths between $BAD_1$ and $A$. At least one of these node disjoint paths does not intersect $S$ which contradicts the definition of $BAD_1$. Similarly, we can bound the size of $BAD_2$, the set of outputs in $SC^n$ which are not reachable from any node in $B$. Given any input $u \not \in BAD_1$ of $SC^1_n$ and any output $v \not \in BAD_2$ of $SC^2_n$ we can argue that $u$ is connected to $v$ in $M_n-S$ since we can reach some node $x \in A$ from $u$ and $v$ can be reached from some node $y \in B$ and {\em any} such pair $x,y$ is connected by a path in $M_n-S$. It follows that $M_n$ is  $(|S|, n - |S|)$-ST-robust.
\end{proof}

\begin{corollary} (of Theorem \ref{is-str}) \label{cor:max-std}
  For all $\epsilon > 0$, there exists a family of DAGs $\mathbb{M} = \{ M_n^\epsilon \}_{n=1}^\infty$, where each $M_n^\epsilon$ is a $c$-maximally ST-robust graphs with $|V(M_n)| \leq c_\epsilon n$, indegree$(M_n) \leq c_\epsilon$, and depth $d = n^{1-\epsilon}$.
\end{corollary}
\begin{proof} (sketch)
  In the proof of Lemma \ref{paths-lemma}, we used $(cn, n^{2/3})$-depth robust graphs. When considering the paths $A_i$ and $B_j$, we were considering connecting the upper half of one path to the lower half of another. Thus, after we remove nodes from $M_n$, there exists a path of length at least $n^{2/3} $ that connects any remaining input to any remaining output. Thus $M_n$ is $c$-maximally ST-robust with depth $d = n^{2/3}$. In \cite{Sch83}, Schnitger provides a construction that is $(cn, n^{1-\epsilon})$-depth robust for all constant $\epsilon > 0.$ By the same arguments we used in this section, we can construct $c$-maximally ST-robust graphs with depth $d = n^{1-\epsilon}$, where the constant $c$ depends on $\epsilon$. 
\end{proof}

\subsection{Constructing Maximal ST-Robust Graphs}

In this section, we construct maximal ST-robust graphs, which are 1-maximally ST-robust, from $c$-maximally ST-robust graphs. We give the following construction:

\begin{construction}[$\mathbb{O}(M_n)$] \label{optSTR}
  Let $M_n$ be a $c$-maximally ST-robust graph on $O(n)$ nodes. Let $O$ be a set $o_1, o_2, \ldots, o_n$ of $n$ output nodes and let $I$ be a set $i_1, i_2, \ldots, i_n$ of $n$ input nodes. Let $S_j$ for $1 \leq j \leq \lceil \frac{1}{c} \rceil$ be a copy of $M_n$ with outputs $o_1^j, o_2^j, \ldots, o_n^j$ and inputs $i_1^j, i_2^j, \ldots, i_n^j$. Then for all $1 \leq j \leq n$ and for all $1 \leq k \leq n$, add a directed edge from $i_k$ to $i_k^j$ and from $o_k^j$ to $o_k$. 
\end{construction}

Because we connect $\lceil \frac{1}{c} \rceil$ copies of $M_n$ to the output nodes, $\mathbb{O}(M_n)$ has indegree $\max \{ \delta, \lceil \frac{1}{c} \rceil\}$, where $\delta$ is the indegree of $M_n$. Also, if $M_n$ has $kn$ nodes, then $\mathbb{O}(M_n)$ has $(k \lceil \frac{1}{c} \rceil + 2)n$ nodes. We now show that $\mathbb{O}(M_n)$ is a maximal ST-robust graph.

\begin{theorem}\label{thm:cmax-str}
  Let $M_n$ be a $c$-maximally ST-robust graph. Then $\mathbb{O}(M_n)$ is a maximal ST-robust graph.
\end{theorem}
\begin{proof}
  Let $R \subset V(\mathbb{O}(M_n))$ with $|R| = k$. Let $R = R_I \cup R_M \cup R_O$, where $R_I = R \cap I$, $R_O = R \cap O$, and $R_M = R \cap \left( \cup_{i = 1}^{\lceil 1/c \rceil} S_i \right)$. Consider $\mathbb{O}(M_n) - R$. We see that $|R_M| \leq k$, so by the Pidgeonhole Principal at least one $S_j$ has less than $cn$ nodes removed, say it has $t$ nodes removed for $t \leq cn$. Hence $t \leq |R_M|$. Since $S_j$ is $c$-max ST-robust there exists a subgraph $H$ of $S_j \ R$ containing $n-t$ inputs and $n-t$ outputs such that every input is connected to all of the outputs. Let $H'$ be the subgraph induced by the nodes in $V(H) \cup I' \cup O'$, where $I' = \{(i_a, i_a^b) | i_a^b \in H\}$ and $O' = \{(o_a^b, o_a) | o_a^b \in H \}$. 

  \begin{claim}
    The graph $H'$ contains at least $n-k$ inputs and $n-k$ outputs and there is a path between every pair of input and output nodes.
  \end{claim}
  \begin{proof}
    The set $|I \setminus I'| \leq |I \cap R| + |V(S_j) \cap R| \leq |R| \leq k$. Similarly, $|O \setminus O'| \leq |O \cap R| + |V(S_j) \cap R| \leq |R| \leq k.$
    Let $v \in I'$ be some input. By the connectivity of $H$, $v$ can reach all of the outputs in $O'$. Thus there is a path between every pair of input and output nodes.
  \end{proof}
  Thus $\mathbb{O}(M_n)$ is $(k, n-k)$-ST-robust for all $1 \leq k \leq n$. Therefore $\mathbb{O}(M_n)$ is a maximal ST-robust graph.
\end{proof}

\begin{corollary} (of Theorem \ref{thm:cmax-str})
  For all $\epsilon > 0$, there exists a family $\mathbb{M^\epsilon} = \{ M_k^\epsilon\}_{k = 1}^\infty$ of max ST-robust graphs of depth $d = n^{1 - \epsilon}$ such that $|V(M_k^\epsilon)| \leq c_\epsilon n$ and indegree$(M_k^\epsilon) \leq c_\epsilon$. 
\end{corollary}
\begin{proof}
  Apply Construction \ref{optSTR} to the family graphs  $\mathbb{M^\epsilon} = \{ M_k^\epsilon\}_{k = 1}^\infty$ from Corollary \ref{cor:max-std}. Then by Theorem \ref{thm:cmax-str}, the family of graphs $\{ \mathbb{O}(M_k^\epsilon) \}_{k=1}^\infty$ is the desired family.
\end{proof}

\section{Applications of ST-Robust Graphs}
As outlined previously maximally ST-Robust graphs give us a tight connection between edge-depth robustness and node-depth robustness. Because edge-depth-robust graphs are often easier to design than node-depth robust graphs~\cite{Sch83} this gives us a fundamentally new approach to construct node-depth-robust graphs. Beyond this exciting connection we can also use ST-robust graphs to construct perfectly tight proofs of space~\cite{ITCS:Pietrzak19a,EC:Fisch19} and asymptotically superior wide-block labeling functions \cite{C:CheTes19}.

\subsection{Tight Proofs of Space}

In Proof of Space constructions~\cite{ITCS:Pietrzak19a} we want to find a DAG $G=(V,E)$ with small indegree along with a challenge set $V_C \subseteq V$. Intuitively, the prover will label the graph $G$ using a hash function $H$ (often modeled as a random oracle in security proofs) such that a node $v$ with parents $v_1,\ldots, v_\delta$ is assigned the label $L_v = H(L_{v_1},\ldots,L_{v_\delta})$. The prover commits to storing $L_v$ for each node $v$ in the challenge set $V_C$. The pair $(G,V_C)$ is said to be $(s,t,\epsilon)$-hard if for any subset $S \subseteq V$ of size $|S| \leq s$ at least $(1-\epsilon)$ fraction of the nodes in $V_C$ have depth $\geq t$ in $G-S$ --- a node $v$ has depth $\geq t$ in $G-S$ if there is a path of length $\geq t$ ending at node $v$.  Intuitively, this means that if a cheating prover only stores $s \leq |V_C|$ labels and is challenged to reproduce a random label $L_v$ with $v \in V_C$ that, except with probability $\epsilon$, the prover will need at least $t$ sequential computations to recover $L_v$ --- as long as $t$ is sufficiently large the verifier the cheating prover will be caught as he will not be able to recover the label $L_v$ in a timely fashion. Pietrzak argued that  $(s,t,\epsilon)$-hard graphs translate to secure Proofs of Space in the parallel random oracle model~\cite{ITCS:Pietrzak19a}. 

We want $G$ to have small indegree $\delta(G)$ (preferably constant) as the prover will need $O(N \delta(G))$ steps and we want $|V_C| = \Omega(N)$ and we want $\epsilon$ to be small while $s,t$ should be larger.  Pietrzak~\cite{ITCS:Pietrzak19a} proposed to let $G_{\eps}$ be an $\epsilon$-extreme depth-robust graph with $N' = 4N$ nodes and to let $V_C = [3N+1,4N]$ be the last $N$ nodes in this graph. An  $\epsilon$-extreme depth-robust graph with $N'$ nodes is $(e,d)$-depth robust for any $e+d \leq (1-\epsilon)N$. Such a graph is $(s,N,s/N+ 4\epsilon)$-hard for any $s \leq N$. Alwen et al.~\cite{EC:AlwBloPie18} constructed $\epsilon$-extreme depth-robust graphs with indegree $\delta(G)=O(\log N)$  though the hidden constants seem to be quite large. Thus, it would take time $O(N \log N)$ for the prover to label the graph $G$. We remark that $\epsilon=s/|V_C|$ is the tightest possible bound one can hope for as the prover can always store $s$ labels from the set $V_C$.

We remark that if we take $V_C$ to be any subset of output nodes from a maximally ST-robust graph and overlay and $(e=s,d=t)$-depth robust graph over the input nodes then the resulting graph will be $(s,t,\epsilon=s/|V_C|)$-hard --- optimally tight in $\epsilon$. In particular, given a DAG $G = (V=[N],E)$ with $N$ nodes devine the overlay graph $H_G$ by starting with a maximally ST-Robust graph with $|V|$ inputs $I=\{x_1,\ldots,x_{|V|}\}$ and $|V|$ outputs $O$ then for every directed edge $(u,v)\in E(G)$ we add the directed edge $(x_u,x_v)$ to $E(H_G)$ and we specify a target set $V_C \subseteq O$.  Fisch~\cite{EC:Fisch19} gave a practical construction of $(G,V_C)$ with indegree $O(\log N)$ that is $(s,N,\epsilon = s/N+\epsilon')$-hard. The constant $\epsilon'$ can be arbitrarily small though the number of nodes in the graph scales with $O(N \log 1/\epsilon')$. Utililizing ST-robust graphs we fix $\epsilon' = 0$ without increasing the size of the graph\footnote{As a disclaimer we are not claiming that our construction would be more efficient than ~\cite{EC:Fisch19} for practical parameter settings.}.

\begin{theorem} \label{OverlayPOS}
If $G$ is $(e,d)$-depth robust then the pair $(H_G,V_C)$ specified above is $(s,t=d+1,s/|V_C|)$-hard for any $s \leq e$.
\end{theorem}
\begin{proof}
Let $S$ be a subset of $|S| \leq s$ nodes in $H_G$. By maximal ST-robustness we can find a set $A$ of $N-|S|$ inputs and $B$ of $N-|S|$ outputs such that every pair of nodes $u \in A$ and $v \in B$ are connected in $H_G-S$. We also note since $A$ contains all but $s$ nodes from $G$ that some node $u \in A$ is the endpoint of a path of length $t$ by $(s,t)$-depth-robustness of $G$. Since $u$ is connected to {\em every} node in $B$ this means that every node $v \in B$ is the endpoint of a path of length {\em at least} $t+1$. 
\end{proof}

This result immediately leads to a $(s,N^{1-\eps},s/N)$-hard pair for any $s \leq N$ which the prover can label in $O(N)$ time as the DAG $G$ has constant indegreee. We expect that in many settings $t=N^{1-\eps}$ would be sufficiently large to ensure that a cheating prover is caught with probability $s/N$ after each challenge i.e., if the verifier expects a response within 3 seconds, but it would take longer to evaluate the hash function $H$ $N^{1-\eps}$ sequential times.
\begin{corollary}
For any constant $\epsilon > 0$ there is a constant indegree DAG $G$ with $O(N)$ nodes along with a target set $V_C \subseteq V(G)$ of $N$ nodes such that the pair $(G,V_C)$ is $(s,t=N^{1-\eps}, s/N)$-hard for any $s \leq N$. 
\end{corollary}
\begin{proof} (sketch)
Let $G$ be an $\left(N, N^{1-\epsilon}\right)$-depth robust graph with $N' = O(N)$ nodes and constant indegree from \cite{Sch83}. We can then take $V_C$ to be any subset of $N$ output nodes in the graph $H_G$ and apply Theorem \ref{OverlayPOS}. 
\end{proof}

If one does not want to relax the requirement that $t=\Omega(N)$ then we can provide a perfectly tight construction with $O(N \log N)$ nodes and constant indegree. Since the graph has constant indegree it will take  $O(N \log N)$ work for the prover to label the graph. This is equivalent to ~\cite{ITCS:Pietrzak19a}, but with perfect tightness $\epsilon = s/N$. 
\begin{corollary}
For any constant $\epsilon > 0$ there is a constant indegree DAG $G$ with $N'=O(N \log N)$ nodes along with a target set $V_C \subseteq V(G)$ of $N$ nodes such that the pair $(G,V_C)$ is $(s,t, s/N)$-hard for any $s \leq N$. 
\end{corollary}
\begin{proof} (sketch)
Let $G$ be an $(N, N \log N)$-depth robust graph with $N' = O(N \log N)$ nodes and constant indegree from \cite{CCS:AlwBloHar17}. We can then take $V_C$ to be any subset of $N$ output nodes in the graph $H_G$ and apply Theorem \ref{OverlayPOS}. 
\end{proof}

Finally, if we want to ensure that the graph only has $O(N)$ nodes and $t=\Omega(N)$ we can obtain a perfectly tight construction with indegree $\delta(G)=O(\log N)$. 
\begin{corollary}
For any constant $\epsilon > 0$ there is a DAG $G$ with $O(N)$ nodes and indegree $\delta(G)=O(\log N)$ along with a target set $V_C \subseteq V(G)$ of $N$ nodes such that the pair $(G,V_C)$ is $(s,N, s/N)$-hard for any $s \leq N$. 
\end{corollary}
\begin{proof} (sketch)
Let $G$ be an $(N,N)$-depth robust graph with $N' = 3N$ nodes from \cite{EC:AlwBloPie18}. We can then take $V_C$ to be any subset of $N$ output nodes in the graph $H_G$ and apply Theorem \ref{OverlayPOS}. 
\end{proof}

\subsection{Wide-Block Labeling Functions}
Chen and Tessaro~\cite{C:CheTes19} introduced source-to-sink depth robust graphs as a generic way of obtaining a wide-block labeling function $H_{\delta, W}:\{0,1\}^{\delta W} \rightarrow \{0,1\}^W$ from a small-block function $H_{fix}:\{0,1\}^{2L}\rightarrow \{0,1\}^L$ (modeled as an ideal primitive). In their proposed approach one transforms a graph $G$ with indegree $\delta$ and into a new graph $G'$ by replacing every node with a source-to-sink depth-robust graph. Labeling a graph $G$ with a wide-block labeling function is now equivalent to labeling $G'$ with the original labeling function $H_{fix}$. The formal definition of Source-to-Sink-Depth-Robustness is presented below:

\begin{definition}[\bf{Source-to-Sink-Depth-Robustness (SSDR)} \cite{C:CheTes19}]
A DAG $G = (V,E)$ is $(e,d)$-source-to-sink-depth-robust (SSDR) if and only if for any $S \subset V$ where $|S| \leq e$, $G - S$ has a path (with length at least $d$) that starts from a source node of $G$ and ends up in a sink node of $G$.
\end{definition}

 If $G$ is $(e,d)$-depth robust and $G'$  is constructed by replacing every node $v$ in $G$ with a $(e^*,d^*)$-source-to-sink-depth-robust (SSDR) and orienting incoming (resp. outgoing) edges into the sources (resp. out of the sinks) then the graph $G'$ is $(ee^*,dd^*)$-depth robust \cite{C:CheTes19} and has cumulative pebbling complexity at least $ed(e^*d^*)$~\cite{EC:AlwBloPie17}. The SSDR graphs constructed in\cite{C:CheTes19} are $(\frac{K}{4}, \frac{\delta K^2}{2})$-SSDR with  $O(\delta K^2)$ vertices and constant indegree. They fix $K:=W/L$ as the ratio between the length of outputs for  $H_{\delta, W}:\{0,1\}^{\delta W} \rightarrow \{0,1\}^W$ and the ideal primitive  $H_{fix}$. Their graph has $\delta K$ source nodes for a tunable parameter $\delta \in \mathbb{N}$,  $O(\delta K^2)$ vertices and constant indegree. Ideally, since we are increasing the number of nodes by a factor of $\delta K^2$ we would like to see the cumulative pebbling complexity increase by a quadratic factor of $\delta^2 K^4$.  Instead, if we start with an $(e,d)$-depth robust graph with cumulative pebbling complexity $O(ed)$ their final graph $G'$ has cumulative pebbling complexity $ed \times \frac{\delta K^3}{8}$. Chen and Tessaro left the problem of finding improved source-to-sink depth-robust graphs as an open research question.

Our construction of ST-robust graphs can asymptotically\footnote{While we improve the asymptotic performance we do not claim to be more efficient for practical values of $\delta, K$.} improve some of their constructions, specifically their constructions of source-to-sink-depth-robust graphs and wide-block labeling functions.

\begin{theorem}
  \label{thm:stsdr}
Let $G$ be a maximal ST-robust graph with depth $d$ and $n$ inputs and outputs. Then $G$ is an $(n-1, d)$-SSDR graph.
\end{theorem}
\begin{proof}
  By the maximal ST-robustness property, $n-1$ arbitrary nodes can be removed from $G$ and there will still exist at least one input node that is connected to at least one output node. Since $G$ has depth $d$, the path between the input node and output node must have length at least $d$.
\end{proof}
By applying Theorem \ref{thm:stsdr} to the construction in Corollary \ref{thm:cmax-str}, we can construct a family of $(\delta K, (\delta K)^{1-\epsilon})$-SSDR graphs with $O(\delta K)$ nodes and constant indegree and $\delta K$ sources. In this case the cumulative pebbling complexity of our construction would be already be $ed \times \delta^2 K^{2-\epsilon}$ which is much closer to the quadratic scaling that we would ideally like to see. We are off by just $K^{\epsilon}$ for a constant $\epsilon >0$ that can be arbitrarily small. To make the comparison easier we could also applying Theorem \ref{thm:stsdr} to obtain a family of $(\delta K^2, (\delta K^{2})^{1-\epsilon})$-SSDR graphs with $O(\delta K^{2})$-nodes and constant indegree. While the size of the SSDR matches~\cite{C:CheTes19}  our new graph is $(e \delta K^2, d (\delta K^2)^{1-\epsilon})$-depth robust and has cumulative pebbling complexity $ed \times \delta^{2-\epsilon} K^{4-2\epsilon} \gg ed \delta K^3$.


\bibliography{bounded-parallel-mhf,password,abbrev3,crypto,some-bib}
\appendix

\section*{Appendix}

\section{Connector Graphs}
 We say that a directed acyclic graph $G = (V, E)$ with $n$ input vertices and $n$ output vertices is an \textbf{\textit{n}-connector} if for any {\em ordered list} $x_1,\ldots,x_r$ of $r$ inputs and any {\em ordered list} $y_1,\ldots,y_r$ of $r$ outputs, $1 \leq r \leq n$, there are $r$ vertex-disjoint paths in $G$ connecting input node $x_i$ to output node $y_i$ for each $i \leq r$. 

\subsection{Connector Graphs are ST-Robust}
We remarked in the paper that any \textbf{\textit{n}-connector} is maximally ST-robust. 
\begin{remindertheorem}{Theorem \ref{connector-thm}}
\connectorIsSTRobust
\end{remindertheorem}
\begin{proofof}{Theorem \ref{connector-thm}}
  Let $D \subseteq V(G)$ with $|D| = k$. Consider $G - D$. Let $A = \{ (s_1, t_1), \ldots, (s_m, t_m) \}$, where the input $s_i \in S$ is disconnected from the output $t_i \in T$ in $G - D$, or $s_i \in D$ or $t_i \in D$. Let $B = \emptyset$.

  Perform the following procedure on $A$ and $B$: Pick any pair $(s_p, t_p) \in A$ and add $s_p$ and $t_p$ to $B$. Then remove the pair from $A$ along with any other pair in $A$ that shares either $s_p$ or $t_p$. Continue until $A$ is empty.

  If we consider the nodes of $B$ in $G$, then there are $|B|$ vertex-disjoint paths between the pairs in $B$ by the connector property, and in $G - D$ at least one vertex is removed from each path. Thus $|B| \leq k$, or we have a contradiction.

  If $(s, t) \in G - (D \cup B)$ are an input to output pair, and $s$ is disconnected from $t$, then by the definition of $A$ and $B$ we would have a contradiction, since $(s, t)$ would still be in $A$. Thus all of the remaining inputs in $G - (D \cup B)$ are connected to all the remaining outputs. 

  Hence, if we let $H = G - (D \cup B)$, then $H$ is a subgraph of $G$ with at least $n - k$ inputs and $n - k$ outputs, and there is a path going from each input of $H$ to each of its outputs. Therefore, $G$ is $(k, n- k)$-ST-robust for all $1 \leq k \leq n$.
\end{proofof}

\paragraph{Butterfly Graphs}
A well known family of constant indegree $n$-connectors, for $n = 2^k$, are the $k$-dimensional butterfly graphs $B_k$, which are formed by connecting two FFT graphs on $n$ inputs back to back. 
By Theorem \ref{connector-thm}, the butterfly graph is also a maximally ST-robust graph. However, the butterfly graph has $\Omega(n \log n)$ nodes and does not yield a ST-robust graph of linear size.   Since $B_k$ has $O(n \log n)$ vertices and indegree of 2, a natural question to ask is if there exists $n$-connectors with $O(n)$ vertices and constant indegree.

\begin{figure}[t]
  \centering
  \label{blfy}
  \begin{tikzpicture}
    [place/.style={circle, draw=black, fill=black, inner sep=0pt, minimum size=2mm}]

    \node[place] (a0) at (0,0) {};
    \node[place] (a1) at (0,1) {};
    \node[place] (a2) at (0,2) {};
    \node[place] (a3) at (0,3) {};
    \node[place] (a4) at (0,4) {};
    \node[place] (a5) at (0,5) {};
    \node[place] (a6) at (0,6) {};
    \node[place] (a7) at (0,7) {};
    
    \node[place] (b0) at (2,0) {}
    edge [-] node[auto] {} (a4)
    edge [-] node[auto] {} (a0);
    \node[place] (b1) at (2,1) {}
    edge [-] node[auto] {} (a5)
    edge [-] node[auto] {} (a1);
    \node[place] (b2) at (2,2) {}
    edge [-] node[auto] {} (a6)
    edge [-] node[auto] {} (a2);
    \node[place] (b3) at (2,3) {}
    edge [-] node[auto] {} (a7)
    edge [-] node[auto] {} (a3);
    \node[place] (b4) at (2,4) {}
    edge [-] node[auto] {} (a0)
    edge [-] node[auto] {} (a4);
    \node[place] (b5) at (2,5) {}
    edge [-] node[auto] {} (a1)
    edge [-] node[auto] {} (a5);
    \node[place] (b6) at (2,6) {}
    edge [-] node[auto] {} (a2)
    edge [-] node[auto] {} (a6);
    \node[place] (b7) at (2,7) {}
    edge [-] node[auto] {} (a3)
    edge [-] node[auto] {} (a7);

    \node[place] (c0) at (4,0) {}
    edge [-] node[auto] {} (b2)
    edge [-] node[auto] {} (b0);
    \node[place] (c1) at (4,1) {}
    edge [-] node[auto] {} (b3)
    edge [-] node[auto] {} (b1);
    \node[place] (c2) at (4,2) {}
    edge [-] node[auto] {} (b0)
    edge [-] node[auto] {} (b2);
    \node[place] (c3) at (4,3) {}
    edge [-] node[auto] {} (b1)
    edge [-] node[auto] {} (b3);
    \node[place] (c4) at (4,4) {}
    edge [-] node[auto] {} (b6)
    edge [-] node[auto] {} (b4);
    \node[place] (c5) at (4,5) {}
    edge [-] node[auto] {} (b7)
    edge [-] node[auto] {} (b5);
    \node[place] (c6) at (4,6) {}
    edge [-] node[auto] {} (b4)
    edge [-] node[auto] {} (b6);
    \node[place] (c7) at (4,7) {}
    edge [-] node[auto] {} (b5)
    edge [-] node[auto] {} (b7);

    \node[place] (d0) at (6,0) {}
    edge [-] node[auto] {} (c1)
    edge [-] node[auto] {} (c0);
    \node[place] (d1) at (6,1) {}
    edge [-] node[auto] {} (c0)
    edge [-] node[auto] {} (c1);
    \node[place] (d2) at (6,2) {}
    edge [-] node[auto] {} (c3)
    edge [-] node[auto] {} (c2);
    \node[place] (d3) at (6,3) {}
    edge [-] node[auto] {} (c2)
    edge [-] node[auto] {} (c3);
    \node[place] (d4) at (6,4) {}
    edge [-] node[auto] {} (c5)
    edge [-] node[auto] {} (c4);
    \node[place] (d5) at (6,5) {}
    edge [-] node[auto] {} (c4)
    edge [-] node[auto] {} (c5);
    \node[place] (d6) at (6,6) {}
    edge [-] node[auto] {} (c7)
    edge [-] node[auto] {} (c6);
    \node[place] (d7) at (6,7) {}
    edge [-] node[auto] {} (c6)
    edge [-] node[auto] {} (c7);

    \node[place] (e0) at (8,0) {}
    edge [-] node[auto] {} (d2)
    edge [-] node[auto] {} (d0);
    \node[place] (e1) at (8,1) {}
    edge [-] node[auto] {} (d3)
    edge [-] node[auto] {} (d1);
    \node[place] (e2) at (8,2) {}
    edge [-] node[auto] {} (d0)
    edge [-] node[auto] {} (d2);
    \node[place] (e3) at (8,3) {}
    edge [-] node[auto] {} (d1)
    edge [-] node[auto] {} (d3);
    \node[place] (e4) at (8,4) {}
    edge [-] node[auto] {} (d6)
    edge [-] node[auto] {} (d4);
    \node[place] (e5) at (8,5) {}
    edge [-] node[auto] {} (d7)
    edge [-] node[auto] {} (d5);
    \node[place] (e6) at (8,6) {}
    edge [-] node[auto] {} (d4)
    edge [-] node[auto] {} (d6);
    \node[place] (e7) at (8,7) {}
    edge [-] node[auto] {} (d5)
    edge [-] node[auto] {} (d7);

    \node[place] (f0) at (10,0) {}
    edge [-] node[auto] {} (e4)
    edge [-] node[auto] {} (e0);
    \node[place] (f1) at (10,1) {}
    edge [-] node[auto] {} (e5)
    edge [-] node[auto] {} (e1);
    \node[place] (f2) at (10,2) {}
    edge [-] node[auto] {} (e6)
    edge [-] node[auto] {} (e2);
    \node[place] (f3) at (10,3) {}
    edge [-] node[auto] {} (e7)
    edge [-] node[auto] {} (e3);
    \node[place] (f4) at (10,4) {}
    edge [-] node[auto] {} (e0)
    edge [-] node[auto] {} (e4);
    \node[place] (f5) at (10,5) {}
    edge [-] node[auto] {} (e1)
    edge [-] node[auto] {} (e5);
    \node[place] (f6) at (10,6) {}
    edge [-] node[auto] {} (e2)
    edge [-] node[auto] {} (e6);
    \node[place] (f7) at (10,7) {}
    edge [-] node[auto] {} (e3)
    edge [-] node[auto] {} (e7);

  \end{tikzpicture}
  \caption{The butterfly graph $B_3$ is both an 8-superconcentrator and an 8-connector. All edges are directed from left to right.}
\end{figure}
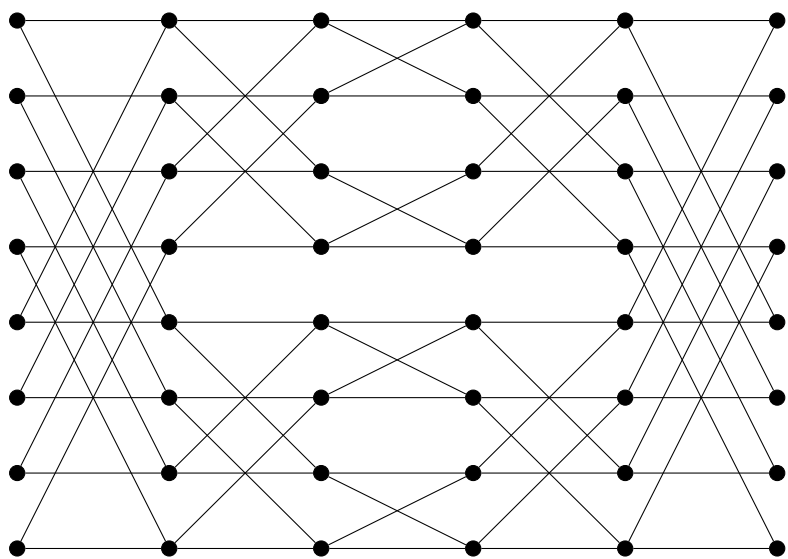

\subsection{Connector Graphs Have $\Omega( n \log n)$ vertices}
An information theoretic argument of  Shannon \cite{shannon} rules out the possibility of linear size $n$-connectors. 

\begin{theorem}{(Shannon~\cite{shannon})}
  An $n$-connector with constant indegree requires at least $\Omega(n \log n)$ vertices.
\end{theorem}

Intuitively, given a $n$-connector with constant indegree with constant indegree and $m$ edges Shannon argued that we can use the $n$-connector to encode any permutatation of $[n]$ using $m$ bits. In more detail fixing any permuation $\pi$ we can find $n$ node disjoint paths from input $i$ to output $\pi(i)$. Because the paths are node disjoint we can encode $\pi$ simply by specifying the subset $S_{\pi}$ of directed edges which appear in one of these node disjoint paths. We require at most $m$ bits to encode $S_{\pi}$ and from $S_{\pi}$ we can reconstruct the set of node-disjoint paths and recover $\pi$. Thus, we must have $m = \Theta(n \log n)$ since we require $\log n! = \Theta(n \log n)$ bits to encode a permutation. 

We stress that this information theoretic argument breaks down if the graph $G$ is only ST-robust. We are guaranteed that $G$ contains a path from input $i$ to output $\pi(i)$, but we are not guaranteed that all of the paths are node disjoint. Thus, $S_{\pi}$ is insufficient to reconstruct $\pi$.

\section{Missing Proofs} \label{apdx:MissingProof}
\begin{remindertheorem}{Theorem \ref{general}}
\thmReductionEdgeToNode
\end{remindertheorem}

\begin{proofof}{Theorem \ref{general}}

  We know that each graph in $\mathbb{M}$ has constant indegree, and that each node $v$ in $G$ will be replaced with a graph in $\mathbb{M}$ with indegree $\text{indeg}\left(M_{\delta(v)}\right)$. Thus $G'$ has maximum indegree $\max_{v \in V(G)}\{ \mathrm{indeg}\left(M_{\delta(v)}\right)\}$. Furthermore, the metanode corresponding to the node $v$ has size $|M_{\delta(v)}|$. Thus $G'$ has $\sum_{v \in V(G)} \left|M_{\delta(v)}\right|$ nodes.
  
  Let $S \subset V(G')$ be a set of nodes that we will remove from $G'$. For a specific node $v \in V(G)$ we let $S_v = S \cap (\{v\} \times V_{\delta(v)})$ denote the subset of nodes deleted from the corresponding metanode. We say that the node $v \in V(G)$ is {\it irrepairable} with respect to $S$ if $|S_v| \geq \delta(v)$; otherwise we say that $v$ is repairable. If a node $v$ is repairable, then because the metanodes are maximally ST-Robust we can find subsets $I_{v,S}$ and $O_{v,S}$ (with $|I_{v,S}|, |O_{v,S}| \geq \delta(v)-|S_v|$) such that each input node $s \in I_{v,S}$ is connected to every output node in $O_{v,S}$.  

 We say that an edge $ (u, v) \in E(G)$ is {\it irrepairable} with respect  if $u$ or $v$ is irrepairable, or if the corresponding edge $e' = (u', v') \in E(G')$ has $u' \not\in O_{u,S}$ or $v' \not\in I_{v, S}$. We let $S_{irr} \subset E(G)$ be the set of irrepairable edges after we remove $S$ from $G$. We begin the proof by first proving two claims. 

  \begin{claim} \label{claim:paths}
    Let $P$ be a path of length $d$ in $G - S_{irr}.$ Then there exists a path of length {\em at least} $d$ in $G' - S.$
  \end{claim}
  \begin{proof}
    In $G - S_{irr}$ we have removed all of the irreparable edges, so any path in the graph contains only repairable edges. By definition, if $(u, v)$ is a repairable edge,  both $u$ and $v$ will be repairable, and $(u, \pi_{out,u}(v)) \in O_{u, S}$ and $(v, \pi_{in,v}(u)) \in I_{v,S}$. Thus the edge corresponding to $(u,v)$ in $G' - S$ will connect the metanodes of $u$ and $v$, and $(u, \pi_{out,u}(v))$ connects to every node in $I_{u, S}$ and $(v, \pi_{out, v}(u))$ connects to every node in $O_{v, S}$. Thus the edges in $G' - S$ corresponding to the edges in $P$ form a path of length at least $d$. 
    \end{proof}

  \begin{claim} \label{claim:irr-bound}
    Let $S_{irr} \subset E(G)$ be the set of irreparable edges with respect to the removed set $S$. Then \[|S_{irr}| \leq 2 |S|.\]
  \end{claim}
  \begin{proof}
If a node $v$ is repairable with respect to $S$ then let $S_{irr,v}^{in} \subseteq E(G)$ (resp. $S_{irr, v}^{out}$) denote the subset of edges $(u,v) \in E(G)$ (resp. $(v,u) \in E(G)$) that are irrepairable because of $S_v$ i.e., the corresponding edge $e' = (u', v') \in E(G')$ has $v' \not \in I_{v, S}$ (resp. the corresponding edge $(v',u') \in E(G')$ has $v' \not \in O_{v,S}$). Let $S_{irr,v} = S_{irr,v}^{in} \cup S_{irr, v}^{out}$. Similarly, if $v$ is irrepairable we let $S_{irr,v} = \{ (u,v)~:(u,v) \in E(G)\} \cup \{(v,u)~:~(v,u) \in E(G)\}$ denote the set of all of $v$'s incoming and outgoing edges. We note that $|S_{irr}| \leq \sum_v \left|S_{irr,v}\right|$ since $S_{irr} = \bigcup_{v} S_{irr,v}$ any irrepairable edge must be in one of the sets $S_{irr,v}$. Now we claim that $|S_{irr,v}| \leq |S_v|$ where $S_v = S \cap (\{v\} \times V_{\delta(v)})$ denote the subset of nodes deleted from the corresponding metanode. We now observe that
 \begin{align*}
      \left|S_{irr,v}\right| & \leq \left|S_{irr,v}^{in}\right| + \left|S_{irr,v}^{in}\right|  \\
& \leq \left( \delta(v) - |I_{v,S}|\right) + \left( \delta(v) - |O_{v,S}|\right) \leq 2 |S_v| \ .
      \end{align*}
The last inequality invokes maximal ST-robustness to show that  $\delta(v) - |O_{v,S}| \leq |S_v|$ and $ \delta(v) - |I_{v,S}| \leq |S_v|$. If a node $v$ is irrepairable then the subsets $I_{v,S}$ and $O_{v,S}$ might be empty since $\delta(v)-|S_v| \leq 0$, but it still holds that $\delta(v) - |O_{v,S}| \leq |S_v|$ and $ \delta(v) - |I_{v,S}| \leq |S_v|$. Thus, we have 
  
    Thus\[
      |S_{irr}|  \leq \sum_{v}  \left|S_{irr,v}\right|  \leq \sum_v 2 |S_v|  \leq 2|S| \ . 
      \]
\end{proof}

\end{proofof}

\begin{remindertheorem}{\corref{DRone}} \thmCorrolaryofDRone
\end{remindertheorem}

\begin{proofof}{\corref{DRone}} (sketch) We slightly modify our reduction.   Instead of replacing each node $v \in G$ with a copy of $M_{\delta(v)}$, we instead replace each node with a copy of $M_{\delta, v} := M_\delta$, attaching the edges same way as in Construction \ref{con:reduce}. Thus the transformed graph $G'$ has $|V(G)|\times \left| M_{\delta}\right|$ nodes and constant indegree. Let $S \subset V(G')$ be a set of nodes that we will remove from $G'$. By Claim \ref{claim:paths}, there exists a path $P$ in $G' - S$ that passes through $d$ metanodes $M_{\delta, v_1},\ldots, M_{\delta, v_d}$. The only difference is that each  $M_{\delta,v_i}$ is maximally ST-robust {\em with depth} $d_\delta$ meaning we can assume that the sub-path $P_i = P \cap M_{\delta,v_i}$ through each metanode has length $|P_i| \geq d_\delta$. Thus, the total length of the path is at least $\sum_i |P_i| \geq d\cdot d_\delta$.
\end{proofof}

\begin{remindertheorem}{Lemma \ref{paths-lemma} \cite{Sch83}} \pathsLemma
\end{remindertheorem}

\begin{proofof}{Lemma \ref{paths-lemma} \cite{Sch83}}
Consider $H_n^1 - S$. Since $H^1_n$ is $(cn, n^{2/3})$-depth-robust and $|S| = cn/2$, there must exist a path $A_1 = (v_1, \ldots, v_{n^{2/3}})$ in $H_n^1 - S$. Remove all vertices of $A_1$ and repeat to find $A_2, \ldots $. Then we finish with $cn / (2n^{2/3}) = cn^{1/3}/2$ vertex disjoint paths of length $n^{2/3}$. We perform the same process on $H^2_n$ to find the $B_i$. 
\end{proofof}

\begin{remindertheorem}{Lemma \ref{prob}}
  \thmProb
\end{remindertheorem}

\begin{proofof}{Lemma \ref{prob}}
 By Lemma \ref{paths-lemma}, we know that in $G_n - S$ there exists $k:=c'n^{1/3}/2$ vertex disjoint paths $A_1, \ldots, A_{k}$ in $H_n^1$ of length $n^{2/3}$ and  $k$ vertex disjoint paths $B_1, \ldots, B_{k}$ in $H_n^2$ of length $n^{2/3}$. Here, $c'$ is the constant from Lemma \ref{paths-lemma}.   Let $U^i_{j,S}$ be the upper half of the $j$-th path in $H_n^i$ and $L_{j,S}^i$ be the lower half, both of which are relative to the  removed set $S$.

Now for each $i \leq k$ define the event $BAD_{i,S}$ to be the event that there exists $j \leq k$ s.t., $U^1_{j,S}$ is disconnected from $L^2_{i,S}$. We now set  $GOOD_S = [k]\setminus\{i\ :\ BAD_{i,S}\}$ and define $$B_S := \bigcup_{i \in GOOD_S}^k U^2_{i,S}~,~~~\mbox{and}~~~~A_S:= \bigcup_{i=1}^k L^1_{i,S} \  . $$  
Now we claim that for every node $u \in A_S$ and $v \in B_S$ the graph $G_n - S$ contains a path from $u$ to $v$. Since $u \in A_S$ we have $u \in L^1_{i,S}$ for some $i \leq k$. Thus, all nodes in $U^1_{i,S}$ are reachable from $u$. Since, $v \in B_S$ we have $v \in U^2_{j,S}$ for some good $j \in GOOD$. We know that $v$ is reachable from any node in $L^2_{j,S}$. By definition of $GOOD_S$ there must be an edge $(x,y)$ from some node $x \in U^1_{i,S}$  to some node $y \in L^2_{j,S}$ since we already know that there is a directed path from $u$ to $x$ and from $y$ to $v$ there is a directed path from $u$ to $v$. Thus, every pair of nodes in $A_S$ and $B_S$ are connected. 

We have $|A_S| \geq k n^{2/3} = c'n/2$. It remains to argue that for any set $S$ the resulting set  $|B_S| = |GOOD_S| n^{2/3}$ is sufficiently large. Now we define the event 

  \begin{align*}
    BAD_{S} & := |\{i\ :\ BAD_{i,S}\}| > \frac{k}{10}.
  \end{align*}
  
 Intuitively, $BAD_{S}$ occurs when more than a small fraction of the events $BAD_{i,S}$ occur. Assuming that $BAD_S$ never occurs then for any set $S$ we have $$|B_S| \geq \left|GOOD_S\right| n^{2/3} \geq (9/10)k n^{2/3} = 9c'n/20 \ . $$ 

Consider, for the sake of finding the probabilities, that $S$ is fixed before all of the random edges are added to $G_n$.  We will then union bound over all choices of sets $S$.  First we consider the probability that a single upper path, say $U^1_{1,S}$ is disconnected from a particular lower path, say $L^2_{1,S}$. There are $n^{1/3}$ possible lower parts to connect to, and there are $n^{2/3}/2$ nodes in the upper part that can connect to the lower part, and there are $\tau$ random edges added from each node in the upper part, so we have that 
  \[ \P \left[U_{1,S}^1 \mathrm{\ disconnected\ from\ } L^2_{1,S} \right] \leq \left(1 - \frac{1}{2n^{1/3}} \right)^{\tau n^{2/3}/2} \leq \left( \frac{1}{e} \right)^{\tau n^{1/3}/4}. \]
  
Union bounding over all indices $j$ we have \[ \P\left[ BAD_{i,S} \right] \leq k  \left( \frac{1}{e} \right)^{\tau n^{1/3}/4 } =  \left( \frac{1}{e} \right)^{\tau n^{1/3}/4 - \ln k} \ .\] 

We remark that for $i \neq j$ the event $BAD_{i,S}$ is independent of $BAD_{j,S}$ since the $\tau$ random incoming edges connected to $L^2_i$ are sampled independently of the edges for $L^2_j$.  

 We will show that the probability of the event $ BAD_{S}$ is very small and then take a union bound over all possible $S$ to show our desired result.
  \begin{align*}
    \P \left[  BAD_{S} \right] & \leq \binom{k}{k/10} \P \left[BAD_{1,S} \wedge \ldots \wedge BAD_{k/10,S} \right] \\
    & = \binom{k}{k/10} \P \left[ BAD^{upper}_{1,S} \right]^{k/10} \\ 
    & \leq \binom{k}{k/10} \left[ \left( \frac{1}{e} \right)^{\tau n^{1/3}/4 - \ln k}  \right]^{\frac{k}{10}} \\
    & = \binom{k}{k/10} \left( \frac{1}{e} \right)^{(k \tau n^{1/3}-4 k \ln k)/40 }\ .
  \end{align*}
   Finally, we take the union bound over every possible $S$ of size $cn/2$ nodes. Since $G_n$ has $2n$ nodes there are at most $2^{2n} = e^{ 2n \ln 2}$ such sets. Thus,
  \begin{align*}
    \P \left[ \exists S \mathrm{\ s.t.\ } BAD_{S} \right] & \leq  e^{ 2n \ln 2}  \P \left[  BAD^{upper}_{S} \right] \leq \left( \frac{1}{e} \right)^{(k \tau n^{1/3}-4 k \ln k)/40 - 2n \ln 2} \ . 
  \end{align*}

By selecting a sufficiently large constant like $\tau = 800/c'$ we can ensure that $(k \tau n^{1/3}-4 k \ln k)/40 - 2n \ln 2 = 20n - 2n \ln 2 - (k \ln k)/10 \geq n$ so that  

\[  \P \left[ \exists S \mathrm{\ s.t.\ } BAD_{S} \right] \leq 2^{-n} \ . \]

Thus, except with negigible probability for any $S$ of size $cn/2$ the event $BAD_S$ does not occur for any set $S$ selected {\em after} $G_n$ is sampled. 
\end{proofof}

\end{document}